\documentclass{llncs}
\usepackage[pdfstartview={FitH}]{hyperref} 
\usepackage[ruled,vlined]{algorithm2e}
\usepackage{relsize}
\toctitle
\tocauthor
\usepackage{etoolbox}
\usepackage{amsmath,amssymb,stmaryrd}
\usepackage{stmaryrd}
\usepackage{amssymb}
\usepackage{float}
\usepackage{cancel}
\DeclareMathOperator{\Pos}{Pos}
\def\regw#1{\Reg{(#1)}}
\DeclareMathOperator{\Reg}{RegExp}
\DeclareMathOperator{\Last}{Last}
\DeclareMathOperator{\last}{Fl_0} 
\DeclareMathOperator{\Follw}{Fl_>}
\DeclareMathOperator{\First}{First}
\DeclareMathOperator{\Follow}{Follow}
\DeclareMathOperator{\Leave}{Leaves}

\DeclareMathOperator{\E}{E}

\DeclareMathOperator{\h}{H}
\DeclareMathOperator{\f}{F}
\DeclareMathOperator{\ar}{R}
\DeclareMathOperator{\G}{G}
\DeclareMathOperator{\Fir}{Fr_>}
\DeclareMathOperator{\Firs}{Fr_0}

\def\b#1{\overline{#1}}
\usepackage{pgf}
\usepackage{tikz}
\usepackage{tikz}
\usetikzlibrary{automata}
\usetikzlibrary{shadows}
\usetikzlibrary{arrows}
\usetikzlibrary{shapes}
\usetikzlibrary{decorations.pathmorphing}
\usetikzlibrary{fit}
\tikzstyle{every picture}=[>=stealth',shorten >=1pt,node distance=1.44cm,bend angle=45,initial text=,every state/.style={inner sep=0.75mm, minimum size=1mm},font=\scriptsize]
\usepackage{theorem}
\usepackage{stmaryrd}
\usepackage{amsxtra}
\usepackage{amssymb}
\usepackage{graphicx}
\DeclareMathOperator{\rooot}{root}
\def\first#1{\First{(#1)}}
\def\Po#1#2{\Pos_{#1}{(#2)}}
\def\firstt#1{\Fir{(#1)}}
\def\firs#1{\Firs{(#1)}}

\def\Fll#1#2#3{\Follow{(#1,#2,#3)}}

 \def\las#1#2#3{\last{(#1,#2,#3)}}
\def\Fw#1#2#3{\Follw{(#1,#2,#3)}}

\title{Algorithm for the $k$-Position Tree Automaton Construction\thanks{D. Ziadi was supported by the MESRS - Algeria under Project 8/U03/7015.}} 
\author{Nadia Ouali Sebti \and Djelloul Ziadi \thanks{\email{$\{$Nadia.Ouali-Sebti, Djelloul.Ziadi$\}$@univ-rouen.fr}}}
\institute{Laboratoire LITIS - EA 4108 Universit\'e de Rouen, Avenue de l'Universit\'e \\76801 Saint-\'Etienne-du-Rouvray Cedex.} 
\authorrunning{N. Ouali Sebti, D. Ziadi}
\date{}
\begin{document}
\setcounter{footnote}{0}
\maketitle             
\begin{abstract}
The word position automaton was introduced by Glushkov and McNaughton in the early $1960$. This automaton
is homogeneous and has $(||\E||+1)$ states for a word expression of alphabetic width $||\E||$.
This kind of automata is extended to regular tree expressions.

In this paper, we give an efficient algorithm that computes the $\Follow$ sets, which are used in different algorithms of conversion of a regular expression into tree automata. In the following, we consider the $k$-position tree automaton construction. We prove that for a regular expression $\E$ of a size $|\E|$ and alphabetic width $||\E||$, the $\Follow$ sets can be computed in $O(||\E||\cdot |\E|)$ time complexity.   
\end{abstract}
\section{Introduction}

This paper is an extended version of ~\cite{Ouali}.

Regular expressions, which are finite representatives of potentially infinite languages, are widely used in various application areas such as XML Schema Languages~\cite{xml}, logic and verification, \emph{etc.} The concept of word regular expressions  has been extended to tree regular expressions.

 In the case of words, it is agreed that each regular expression can be transformed into a non-deterministic finite automaton.
Computer scientists have been interested in designing efficient algorithms for the computation of the position automaton. 
Three well-known algorithms for the computation of this automaton exist. The first makes use 
of the notion of star normal form ~\cite{Brug} of a regular expression   . The second is based on a
lazy computation technique ~\cite{Paig}. The third is built on the so-called ZPC-structure ~\cite{ZPC}.
The complexity of these three algorithms is quadratic with regard to the size of the regular expression.

This study is motivated by the development of a library of functions for handling rational kernels ~\cite{mohri1} in the case of trees. The first problem consists of the conversion of a regular expression into a tree automaton. 

Recently Kuske and Meinecke~\cite{automate2} proposed an Algorithm to construct an equation automaton~\cite{antimirov,ouardi} from a regular tree expression  $\E$  with an $O(\ar\cdot{|\E|}^2)$ time complexity where $|\E|$ is the size of $\E$ and $\ar$ is the maximal rank appearing in the ranked alphabet. 
This algorithm is an adaptation to trees of the one given by Champarnaud and Ziadi in the case of words~\cite{ZPC}. 
This generalization is interesting although the adaptation of the word algorithm to trees is not obvious at all.
Indeed, the Champarnaud and Ziadi Algorithm, for the construction of the set of transitions, is based 
on the computation of some function called "$\Follow$" which is not yet defined on trees.  
Notice that the star normal form of a regular tree expression $\E$ can not be defined, this notion doesn't make sense.  
For these reasons the definition of the $\Follow$ function in the case of trees is given in this paper, while an efficient algorithm for its computation (computation of the $k$-position tree automaton) is proposed. 

The paper is organized as follows: Section~\ref{sec prelim} outlines finite tree automata over ranked alphabets, regular tree expressions, and linearized regular tree expressions. Next, in Section~\ref{sec automata} the notions of $\First$ and $\Follow$ of regular expressions and the $k$-position automaton are recalled. Then, in Section~\ref{sec algo} we present an efficient algorithm which builds the $k$-position tree automaton with an $O(||\E||\cdot|\E|)$ time complexity. Finally, the different results described in this paper are given in the conclusion. 

\section{Preliminaries}\label{sec prelim}
    Let $(\Sigma,\mathrm{r})$ be  \emph{a ranked alphabet}, where $\Sigma$ is a finite set and $\mathrm{r}$ represents the  \emph{rank} of $\Sigma$ which is a mapping from $\Sigma$ into $\mathbb{N}$. The set of symbols of rank $n$ is denoted by $\Sigma_{n}$. The elements of rank $0$ are called  \emph{constants}. A \emph{tree} $t$ over   $\Sigma$ is inductively defined as follows: $t=a,~ t=f(t_1,\dots,t_k)$ where $a$ is any symbol in  $\Sigma_0$, $k$ is any integer satisfying $k\geq 1$, $f$ is any symbol in $\Sigma_k$ and $t_1,\dots,t_k$ are any $k$ trees over $\Sigma$. We denote by $T_{\Sigma}$ the set of trees over $\Sigma$.  \emph{A tree language} is a subset of $T_{\Sigma}$. Let $\Sigma_{>}=\Sigma\backslash \Sigma_0$ denote the set of  \emph{non-constant symbols} of the ranked alphabet $\Sigma$. \emph{A Finite Tree Automaton} (FTA)~~\cite{automate1,automate2} ${\cal A}$ is a tuple $( Q, \Sigma, Q_{T},\Delta )$ where $Q$ is a finite set of states, $Q_T \subset Q$ is the set of \emph{final states} 
and  $\Delta\subset\bigcup_{n\geq 0}(Q \times \Sigma_{n}\times Q^n)$ is the set of  \emph{transition rules}. This set is equivalent to the function $\Delta$ from $Q^n \times \Sigma_n$ to $2^Q$ defined  by $(q,f,q_1,\dots,q_n)\in \Delta\Leftrightarrow q\in \Delta(q_1,\dots,q_n,f)$. The domain of this function can be extended to $(2^Q)^n \times \Sigma_n$
as follows: $\Delta(Q_1,\dots,Q_n,f)=\bigcup_{(q_1,\dots,q_n)\in Q_1\times\dots\times Q_n} \Delta(q_1,\dots,q_n,f)$.  Finally, we denote by $\Delta^*$ the function from  $T_{\Sigma}\rightarrow 2^Q$  defined for any tree in $T_{\Sigma}$ as follows: 
  $\Delta^*(t)= \Delta(a)$ if  $t=a$ with $a\in \Sigma_0$, $\Delta^*(t)=\Delta(\Delta^*(t_1),\dots,\Delta^*(t_n),f)$
   if $t=f(t_1,\dots,t_n)$ with $f\in {\Sigma}_n$ and $t_1,\ldots,t_n\in T_{\Sigma}$. 
  A tree is \emph{accepted} by ${\cal A}$ if and only if $\Delta^*(t)\cap Q_T\neq \emptyset$.   
 
  The language ${\cal L(A)}$ \emph{recognized} by $A$ is the set of trees accepted by ${\cal A}$  \emph{i.e.} ${\cal L(A)}=\{t\in T_{\Sigma}\mid \Delta^*(t)\cap Q_T\neq \emptyset\}$. 

For any integer $n\geq 0$, for any $n$ languages $L_1, \dots, L_n\subset T_{\Sigma}$, and for any symbol  $f\in \Sigma_n$, $f(L_1, \dots, L_n)$ is the tree language $\lbrace f(t_1, \dots, t_n)\mid t_i\in L_i\rbrace$. The \emph{tree substitution} of a constant $c$ in $\Sigma$ by a language $L\subset T_{\Sigma}$ in a tree $t\in T_{\Sigma}$, denoted by $t\lbrace c \leftarrow L\rbrace$, is the language inductively defined by:
   $L$ if $t=c$; $\lbrace d\rbrace$ if $t=d$ where $d\in \Sigma_0\setminus\{c\}$; $f(t_1\lbrace c \leftarrow L\rbrace, \dots, t_n\lbrace c \leftarrow L\rbrace)$ if $t=f(t_1, \dots, t_n)$ with $f\in\Sigma_n$ and $t_1, \dots, t_n$ any $n$ trees over $\Sigma$.   
    Let $c$ be a symbol in $\Sigma_0$. The $c$-\emph{product} $L_1\cdot_{c} L_2$ of two languages $L_1, L_2\subset T_{\Sigma}$ is  defined by $L_1\cdot_{c} L_2=\bigcup_{t\in L_1}\lbrace t\lbrace c \leftarrow L_2\rbrace \rbrace$. The \emph{iterated $c$-product} is  inductively  defined for $L\subset T_{\Sigma}$ by:  $L^{0_c}=\lbrace c \rbrace$ and $L^{{(n+1)}_c}=L^{n_c}\cup L\cdot_{c} L^{n_c}$. The $c$-\emph{closure} of $L$ is defined by  $L^{*_c}=\bigcup_{n\geq 0} L^{n_c}$.
    
A \emph{regular expression } over a ranked alphabet $\Sigma$ is inductively defined by  $\E=0$, $\E\in \Sigma_0$, $\E=f(\E_1, \cdots, \E_n)$, $\E=(\E_1+\E_2)$, $\E=(\E_1\cdot_c \E_2)$, $\E=({\E_1}^{*_c})$, where $c\in\Sigma_0$, $n\in\mathbb{N}$, $f\in\Sigma_n$ and $\E_1,\E_2 ,\dots, \E_n$ are any $n$ regular expression s over $\Sigma$. Parenthesis can be omitted when there is no ambiguity. We write $\E_1=\E_2$ if $\E_1$ and $\E_2$ graphically coincide. We denote by $\regw{\Sigma}$ the set of all regular expression s over $\Sigma$. Every regular expression  $\E$ can be seen as a tree over the ranked alphabet $\Sigma\cup \{+,\cdot_c, *_c \mid c\in \Sigma_0\}$ where $+$ and $\cdot_c$ can be seen as symbols
 of rank $2$ and $*_c$ has rank $1$. This tree is the syntax-tree $T_{\E}$ of $\E$. We denote by ${|\E|}_f$ the number of occurrences of a symbol $f$ in a regular expression $\E$. The \emph{alphabetic width} $||\E||$ of $\E$ is the number of occurrences of symbols of $\Sigma_{>}$ in $\E$ ( $||\E||=(\sum_{f\in\Sigma_{>}}{|\E|}_f$). \emph{The size} $|\E|$ of $\E$ is the size of its syntax tree $T_{\E}$. The \emph{language} $\llbracket \E\rrbracket$  \emph{denoted by} $\E$ is inductively defined by
 $\llbracket 0\rrbracket=\emptyset$, $\llbracket c\rrbracket=\lbrace c\rbrace$, $\llbracket f(\E_1,\E_2 , \cdots, \E_n)\rrbracket= f(\llbracket \E_1 \rrbracket, \dots, \llbracket \E_n \rrbracket)$, $\llbracket \E_1+ \E_2\rrbracket=\llbracket \E_1\rrbracket\cup\llbracket \E_2 \rrbracket$, $\llbracket \E_1\cdot_{c} \E_2\rrbracket=\llbracket \E_1\rrbracket \cdot_{c}\llbracket \E_2 \rrbracket$, $\llbracket {\E_1}^{*_c}\rrbracket=\llbracket \E_1\rrbracket^{*_c}$ where  $n\in\mathbb{N}$, $\E_1,\E_2,\dots, \E_n$ are any $n$ regular expression s, $f\in\Sigma_n$  and $c\in \Sigma_0$. It is well known that a tree language  is accepted by some tree automaton if and only if it can be denoted by a regular expression  ~\cite{automate1,automate2}.
A regular expression  $\E$ defined over $\Sigma$ is  \emph{linear} 
if every symbol of rank greater than $1$ appears at most once in $\E$. Note that any constant symbol may occur more than once. Let $\E$ be a regular expression  over $\Sigma$. The  \emph{linearized regular expression } $\b\E$ in $\E$ of a regular expression  $\E$ is obtained from $\E$ by marking differently all symbols of a rank  greater than or equal to $1$ (symbols of $\Sigma_>$). The marked symbols form together with the constants in $\Sigma_0$ a ranked alphabet $\mathrm{Pos}_E(E)$ the symbols of which we call \emph{positions}.
The mapping $h$ is defined from $\Po{\E}{\E}$ to $\Sigma$ with $h(\Po{\E}{\E}_m)\subset \Sigma_m$ for every  $m\in \mathbb{N}$. It associates with a marked symbol $f_j\in  \Po{\E}{\E}_{>}$ the symbol $f\in \Sigma_>$ and for a symbol $c\in \Sigma_0$ the symbol $h(c)=c$.
We can extend the mapping $h$ naturally to  $\regw{\Po{\E}{\E}}\rightarrow\regw{\Sigma}$ by $h(a)=a$, $h(\E_1+\E_2)=h(\E_1)+h(\E_2)$, $h(\E_1\cdot_c\E_2)=h(\E_1)\cdot_c h(\E_2)$, $h(\E_1^{*_c})=h(\E_1)^{*_c}$, $h(f_j(\E_1,\dots,\E_n))=f(h(\E_1),\dots,h(\E_n))$, with $n\in\mathbb{N}$, $a\in \Sigma_0$, $f\in \Sigma_n$, $f_j\in \Po{\E}{\E}_n$ such that $h(f_j)=f$ and $\E_1,\dots,\E_n$ any regular expression s over $\Po{\E}{\E}$.

\section{The $k$-Position Tree Automaton}\label{sec automata}

The set of positions associated to $\E$ are straightforwardly deduced from the set of symbols associated to $\E$. 
In order to construct a non$-$deterministic finite automaton (position tree automaton) associated to the 
regular expression  $\E$ that recognizes $\llbracket \E\rrbracket$,
we need to define two sets, the set $\first{\b\E}$ and the set $\Follow((\b\E,f_j,k)$ for a position $f_j\in \Po{\E}{\E}_{>}$.

  In the following of this section, $\E$ is a regular expression  over a ranked alphabet $\Sigma$. 
  The set of symbols in $\Sigma$ that appear in an expression $\f$ is denoted by $\Sigma^{\f}$.

  In this section, we show how to compute the $k$-position tree automaton of a regular expression   $\E$, recognizing $\llbracket \E\rrbracket$. This is an extension of the well-known position automaton~~\cite{glushkov} for word regular expression  s where the $k$ represents the fact that any $k$-ary symbol is no longer a state of the automaton, but is exploded into $k$ states.
  The same method was presented independently by McNaughton and Yamada~~\cite{mcnaughton60}.  
   Its computation is based on the computations of particular \emph{position functions}, defined in the following.
  
  In what follows, for any two trees $s$ and $t$, we denote by $s \preccurlyeq t$ the relation "$s$ is a subtree of $t$".
  Let $t=f(t_1,\dots,t_n)$ be a tree. 
We denote by $\rooot(t)$ the root of $t$, by $k\mbox{-}\mathrm{child(t)}$ the $k^{th}$ child of $f$ in $t$, that is the root of $t_k$ if it exists, and by $\Leave(t)$ the set of the leaves of $t$, \emph{i.e.} $\{s\in \Sigma_0\mid s\preccurlyeq t\}$. 
  We denote by $\rooot(t)$ the root of $t$, by $k\mbox{-}\mathrm{child(t)}$ the $k^{th}$ child of $f$ in $t$, that is the root of $t_k$ if it exists, and by $\Leave(t)$ the set of the leaves of $t$, \emph{i.e.} $\{s\in \Sigma_0\mid s\preccurlyeq t\}$. 
  
  Let $\E$ be a regular expression  and $\b\E$ its linearized form, $1\leq k\leq m$ be two integers and $f$ be a symbol in $\Sigma_m$ and $f_j$ be a position in $\Po{\E}{\E}_m$ with $h(f_j)=f$.

  The set $\First(\b\E)$ is the subset of $\Po{\E}{\E}$ defined by $\{\rooot(t)\in \Po{\E}{\E} \mid t\in \llbracket\b\E \rrbracket\}$; The set $\Follow(\b\E,f_j,k)$ is the subset of $\Po{\E}{\E}$ defined by $\{g_i\in \Po{\E}{\E} \mid \exists t\in \llbracket \b\E \rrbracket, \exists s\preccurlyeq t, \mathrm{root}(s)=f, k\mbox{-}\mathrm{child(s)}=g_i\}$; The set $\Last(\b\E)$ is the subset of $\Po{\E}{\E}_0$ defined by $\Last(\b\E)=\displaystyle\bigcup_{t\in\llbracket \b\E\rrbracket}\Leave(t)$.
  
 \begin{example}\label{Pos Automat}
 Let $\Sigma=\Sigma_0\cup\Sigma_1\cup\Sigma_2$ be defined by $\Sigma_0=\{a,b,c\}$, $\Sigma_1=\{f,h\}$ and $\Sigma_2=\{g\}$.
    Let us consider the regular expression   $\E$ and its linearized form defined by:

\noindent$\E=(f(a)^{*_a}\cdot_a b+ h(b))^{*_b}+g(c,a)^{*_c}\cdot_c (f(a)^{*_a}\cdot_a b+ h(b))^{*_b}$,

\noindent$\b\E=(f_1(a)^{*_a}\cdot_a b+ h_2(b))^{*_b}+g_3(c,a)^{*_c}\cdot_c (f_4(a)^{*_a}\cdot_a b+ h_5(b))^{*_b}$.
      
The language denoted by $\b\E$ is $\llbracket \b\E\rrbracket=\{b, f_1(b),f_1(f_1(b)),f_1(h_2(b)),h_2(b),\\
h_2(f_1(b)),h_2(h_2(b)), \ldots,g_3(b,a), g_3(g_3(b,a),a),g_3(f_4(b),a),g_3(h_5(b),a) ,f_4(f_4(b)),\\ f_4(h_5(b),h_5(f_4(b)),h_5(h_5(b)),\ldots\}$.

Consequently, $\First(\b\E)=\{b,f_1,h_2,g_3,f_4,h_5\}$ and $\Follow(\b\E,f_1,1)=\{b,f_1,h_2\}$, $\Follow(\b\E,h_2,1)=\{b,f_1,h_2\}$, $\Follow(\b\E,g_3,1)=\{b,g_3,f_4,h_5\}$, $\Follow(\b\E,g_3,2)=\{a\}$, $\Follow(\b\E,f_4,1)=\{b,f_4,h_5\}$,  $\Follow(\b\E,h_5,1)=\{b,f_4,h_5\}$. 
\end{example}

The two functions $\mathrm{First}$ and $\mathrm{Follow}$ are sufficient to construct the \emph{$k$-position tree automaton} from a regular expression  $\E$. 
 
  \begin{definition}\label{def aut pos}~\cite{arxiv}
    Let $\E$ be a regular expression , $f$ and $g$ be symbols in $\Sigma$ and $f_j$ and $g_i$ be positions in $\Po{\E}{\E}$ with $h(f_j)=f$ and $h(g_i)=g$. The $k$-\emph{Position Tree Automaton} ${\cal P_{\E}}$ is the automaton $(Q,\Sigma,Q_T,\Delta)$ defined by 
       
\centerline{$\begin{array}{r@{\ }c@{\ }l}
      Q=  & & \{f^k_j \mid f_j\in \Po{\E}{\E}_{m}\wedge 1\leq k\leq m\}\\
      & \cup & \{\varepsilon^1\}\mbox{ with } \varepsilon^1 \mbox{ a new symbol not in } \Sigma,\; Q_T=\{\varepsilon^1\}\\ 
\Delta = & & \{(f^k_j,h(g_i),g^1_i,\ldots,g^n_i)\mid g_i \in \Follow(\E,f_j,k)\} \\
            & \cup &  \{(\varepsilon^1,h(f_j),f^1_j,\ldots,f^m_j)\mid f_j \in\First(\E)\}\\
      \end{array}$}

  \end{definition}
  
  It has been shown in ~\cite{arxiv} that the $k$-position tree automaton of $\E$ accepts $\llbracket \E\rrbracket$, hence the following theorem:

    \begin{theorem}\label{thm lang pe eq e}~\cite{arxiv}
     Let $\E$ be a regular expression, then ${\cal L}({\cal P}_{\E})=\llbracket \E\rrbracket$.
  \end{theorem}
    
\begin{example}
The $k$-Position Automaton ${\cal  P}_{\E}$ associated with $\E$ of Example~\ref{Pos Automat} is given in Figure~\ref{fig a t e2}.
The set of states is $Q=\{\varepsilon^1,f^1_1,h^1_2,g^1_3,g^2_3,f^1_4,h^1_5\}$. The set of final states is $Q_T=\{\varepsilon^1\}$. 
The set of transition rules $\Delta$ is 

\centerline{
    $f(f_1^1)\rightarrow \varepsilon^1$, $f(f_1^1)\rightarrow f_1^1$, $f(f_1^1)\rightarrow h_2^1$,
}

\centerline{
    $h(h_2^1)\rightarrow \varepsilon^1$, $h(h_2^1)\rightarrow f_1^1$, $h(h_2^1)\rightarrow h_2^1$,
}

\centerline{
    $g(g_3^1,g_3^2)\rightarrow g_3^1$, $g(g_3^1,g_3^2)\rightarrow \varepsilon^1$,
}

\centerline{
    $f(f_4^1)\rightarrow \varepsilon^1$, $f(f_4^1)\rightarrow g_3^1$, $f(f_4^1)\rightarrow f_4^1$, $f(f_4^1)\rightarrow h_5^1$,
}

\centerline{
    $h(h_5^1)\rightarrow \varepsilon^1$, $h(h_5^1)\rightarrow g_3^1$, $h(h_5^1)\rightarrow f_4^1$, $h(h_5^1)\rightarrow h_5^1$,
}

\centerline{
   $a\rightarrow g_3^2$, $b\rightarrow \varepsilon^1$, $b\rightarrow f_1^1$, $b\rightarrow h_2^1$, $b\rightarrow g_3^1$, $b\rightarrow f_4^1$, $b\rightarrow h_5^1$
}

The $k$-Position Automaton ${\cal P}_{\E}$ associated with $\E$ is represented in Figure~\ref{fig a t e2}.
\end{example}

\begin{figure}[H]
  \centerline{
	\begin{tikzpicture}[node distance=2.5cm,bend angle=30,transform shape,scale=1]
	  \node[accepting,state] (eps) {$\varepsilon^1$};
	  \node[state, above left of=eps] (f11) {$f^1_1$};	
	  \node[state, above right of=eps] (h12) {$h^1_2$}; 
      \node[state, below of=eps] (g13) {$g^1_3$};
      \node[state, below right of=eps,node distance=1.5cm] (cer1) {};
      \node[state, left of=cer1,node distance=0.8cm] (cer) {};
      \node[state, right of=g13,node distance=3.5cm] (g23) {$g^2_3$};
	  \node[state, below left of=g13,node distance=3.5cm] (h15) {$h^1_5$};
      \node[state, below right of=g13,node distance=3.5cm] (f14) {$f^1_4$};
	  \draw (eps) ++(-1cm,0cm) node {$b$}  edge[->] (eps);  
	  \draw (f11) ++(-1cm,0cm) node {$b$}  edge[->] (f11);  
	  \draw (h12) ++(1cm,0cm) node {$b$}  edge[->] (h12); 
	  \draw (h15) ++(0cm,-1cm) node {$b$}  edge[->] (h15);  
	  \draw (g23) ++(1cm,0cm) node {$a$}  edge[->] (g23);  
	  \draw (g13) ++(-1cm,0cm) node {$b$}  edge[->] (g13);    
	  \draw (f14) ++(0cm,-1cm) node {$b$}  edge[->] (f14);
      \path[->]
        (f11) edge[->,below left] node {$f$} (eps)
		(f11) edge[->,loop,above] node {$f$} ()
		(h12) edge[->,bend right,above] node {$h$} (f11)
		(g13) edge[->] node {} (cer1)
		(g23) edge[->] node {} (cer1)	
		(cer1) edge[->,bend right,above right] node {$g$} (eps)
		(g13) edge[->] node {} (cer)
		(g23) edge[->] node {} (cer)	
		(cer) edge[->,bend right=60,above right] node {$g$} (g13)		  
		(h12) edge[->,loop,above] node {$h$} ()
		(h12) edge[->,below right] node {$h$} (eps)
		(f11) edge[->,bend right,above] node {$f$} (h12)
		(h15) edge[->, in=135,out=-135,loop,left] node {$h$} ()	
		(h15) edge[->,above left] node {$h$} (eps)		
		(h15) edge[->,above left] node {$h$} (g13)		
		(h15) edge[->,bend right,above] node {$h$} (f14)
		(f14) edge[->,in=45,out=-45,loop,right] node {$f$} ()	
		(f14) edge[->,bend right,above] node {$f$} (h15)		
		(f14) edge[->,above right] node {$f$} (eps)		
		(f14) edge[->,above right] node {$f$} (g13)
	  ;
      \end{tikzpicture}
  }
  \caption{The $k$-Position Automaton ${\cal P}_{\E}$ of $\E=(f(a)^{*_a}\cdot_a b+ h(b))^{*_b}+g(c,a)^{*_c}\cdot_c (f(a)^{*_a}\cdot_a b+ h(b))^{*_b}$.}
  \label{fig a t e2}
\end{figure}
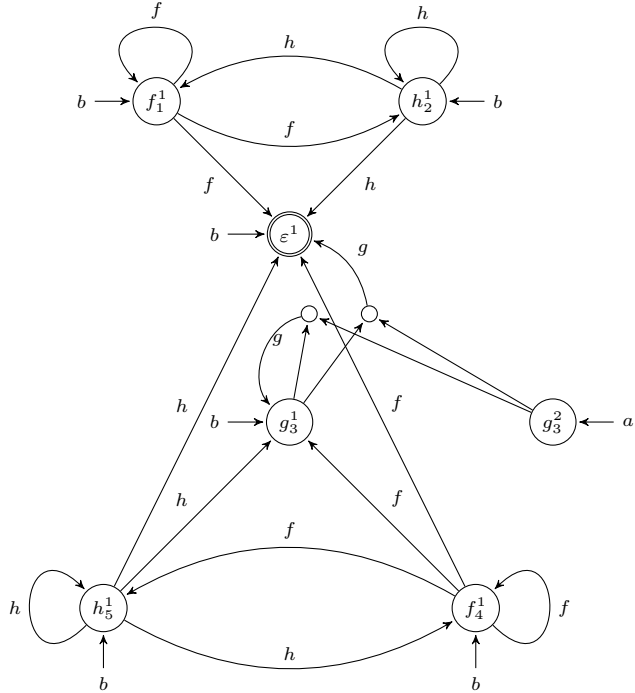	

In the following sections, we will show how we can efficiently compute the function $\Follow(\E,f_j,k)$. This algorithm can be used in different constructions such us the equation automaton~\cite{automate2}, $k$\emph{-C-continuation automaton}~\cite{cie,arxiv} and Follow Automaton~\cite{arxiv}.

\section{Efficient computation of the function $\Follow$}\label{sec algo}
In ~\cite{automate1} Champarnaud and Ziadi gave in the case of words an algorithm with an $O(||\E||\cdot|\E|^2 )$ space and time complexity. 
They enhanced the algorithm to one with an $O(|\E|^2 )$ time and space complexity.
In ~\cite{automate2}, Kuske and Meinecke extend the algorithm based on the notion of word partial derivatives ~\cite{antimirov} to tree partial derivatives in order to compute from a regular expression   $\E$ a tree automaton recognizing $\llbracket \E\rrbracket$. Laugerotte et al. proposed an algorithm for the computation of the position tree automaton and the reduced tree automaton in~\cite{Ouali}. This is an extended version of ~\cite{Ouali}. In~\cite{cie,arxiv2} Mignot et al. gave an efficient algorithm for the computation of the equation automaton using the $k$-c-continuations.

In this section we will describe an algorithm for the computation of the $k$-position tree automaton based on the computation of the $\Follow$ function.

In the following, we will inductively replace each regular subexpression ${\f}^{*_c}$ of $\E$ by the regular subexpression $(\f+c)^{*_c}$. The regular expressions considered thereafter are already dealt by this transformation.

By misuse of language we will denote by $\First(\E)$ for $\First(\b\E)$ and by $\Follow(\E,f_j,k)$ for $\Follow(\b\E,f_j,k)$. Let us first show that the functions $\mathrm{First}$ and $\mathrm{Follow}$ can be inductively computed.
 
  \begin{lemma}\label{firstComput}~\cite{arxiv}
     Let $\E$ be a linear regular expression.  
    The set $\First(\E)$ can be computed as follows:
    
    \centerline{$\First(0)=\emptyset$, $\First(a)=\lbrace a\rbrace$,}
    
\centerline{ $\First(f_j(\E_1, \cdots,\E_m))=\lbrace f_j\rbrace$, 
}    
    
    \centerline{$\First(\E_1+\E_2)=\First(\E_1)\cup\First(\E_2)$,}
    
     \centerline{$\First({\E_1}^{*_c})=\First(\E_1)$,}
    
    \centerline{
      $\First(\E_1\cdot_c \E_2)=
        \left\{
          \begin{array}{l@{\ }l}
            (\First(\E_1)\setminus\{c\}) \cup \First(\E_2) & \text{ if } c\in\llbracket \E_1\rrbracket,\\
            \First(\E_1) & \text{ otherwise.}\\ 
          \end{array}
        \right.
      $
    }
  \end{lemma}  
  
  \begin{lemma}\label{Computfollow}~\cite{arxiv}
   Let $\E$ be a linear regular expression, $1\leq k\leq m$ be two integers and $f_j$ be a symbol in $\Sigma_m$. 
  
    The set of symbols $\Follow(\E,f_j,k)$ can be computed inductively as follows:
    
    \centerline{$\Follow(0,f_j,k)=\Follow(a,f_j,k)=\emptyset$,}
    
    \centerline{
      $\Follow(g_i(\E_1,\ldots,\E_m),f_j,k)=
        \left\{
          \begin{array}{l@{\ }l}
            \First(\E_k) & \text{ if } g_i=f_j,\\
            \Follow(\E_l,f_j,k) & \text{ if } \exists l\mid f_j\in \Sigma^{\E_l},\\
            \emptyset & \text{ otherwise.} 
          \end{array}
        \right.$
    }
    
    \centerline{
      $\Follow(\E_1+\E_2,f_j,k)=
        \left\{
          \begin{array}{l@{\ }l}
            \Follow(\E_1,f_j,k) & \text{ if } f_j\in \Sigma^{\E_1},\\
            \Follow(\E_2,f_j,k) & \text{ if } f_j\in \Sigma^{\E_2},\\ 
             \emptyset & \text{ otherwise.} 
          \end{array}
        \right.$
    }
    
    \centerline{
      $\Follow(\E_1 \cdot_c \E_2,f_j,k)=
        \left\{
          \begin{array}{l@{\ }l}
            (\Follow(\E_1,f_j,k)\setminus\{c\}) \cup \First(\E_2) & \text{ if }  f_j\in \Sigma^{\E_1}\\
             &\ \  \wedge c\in\Follow(\E_1,f_j,k),\\
            \Follow(\E_1,f_j,k) & \text{ if } f_j\in \Sigma^{\E_1}\\
            &\ \  \wedge c\notin\Follow(\E_1,f_j,k),\\
            \Follow(\E_2,f_j,k) & \text{ if } f_j\in \Sigma^{\E_2}\\
            & \ \  \wedge c\in\Last(\E_1),\\
            \emptyset & \text{ otherwise.}
          \end{array}
        \right.$
    }
    
    \centerline{
      $\Follow(\E_1^{*_c},f_j,k)=
        \left\{
          \begin{array}{l@{\ }l}
            \Follow(\E_1,f_j,k) \cup \First(\E_1) & \text{ if } c\in\Follow(\E_1,f_j,k),\\
            \Follow(\E_1,f_j,k) & \text{ otherwise.}\\ 
          \end{array}
        \right.$
    }
  \end{lemma} 
The main idea of our algorithm consists of the separation of the computation of the function $\First$ (resp. $\Follow$) to the computation of two subsets $\Firs$ (resp. $\last$) and $\Fir$ (resp. $\Follw$) that are respectively  
the projection of the set $\First$ (resp. $\Follow$) to the positions associated with
symbols of a rank $0$ and a rank greater than $0$.

Thus the computation of the set $\first{\E}$ can be written as follows:  
 $$\first{\E}=\firs{\E}\uplus\firstt{\E}.$$

\begin{proposition}
 \label{s}
Let $\E$ be a linear regular expression and $\h$ be a subexpression of $\E$. The set of symbols $\firs{\h}$ is defined as follows:
\begin{eqnarray*}
\firs{f_j(\E_1, \cdots,\E_m)}&=&\emptyset,\\
\firs{0}=\emptyset,\ \firs{a}&=& \{a\},\\
\firs{\E_1+\E_2}&=&\firs{\E_1}\cup\firs{\E_2},\\
\firs{{\E_1}\cdot_c \E_2}&=&\left\{
\begin{array}{ll}
(\firs{\E_1}\setminus\{c\})\cup\firs{\E_2}&\mbox{ if  }c\in\llbracket\E_1\rrbracket,\\
\firs{\E_1}&\mbox{ otherwise.}
\end{array}\right.\\
\firs{{\E_1}^{*_c}}&=&\firs{\E_1}.
\end{eqnarray*}
\end{proposition}
\begin{proof}
\begin{sloppy}  
  
Let $\E$ be a linear regular expression, $1\leq k\leq m$ be two integers and $f_j$ be a symbol in $\Sigma^{\E}_m$.
\begin{enumerate}
\item If $\E=0$ or if $\E=f_j(\E_1,\ldots,\E_m)$, then $\firs{\E}=\emptyset$ and for $\E=a$, $\firs{\E}=\{a\}$.  

\item Let us prove this proposition for the case $\E=\E_1\cdot_c \E_2$. 

We have $\firs{\E_1\cdot_c \E_2}=\First(\E_1\cdot_c \E_2,f_j,k)\cap{\Sigma}_0$
 \begin{align*} 
\firs{\E_1\cdot_c \E_2}&=\First(\E_1\cdot_c \E_2)\cap{\Sigma}_0\\ 
&=\left\{
          \begin{array}{l@{\ }l}
            ((\First(\E_1)\setminus\{c\})\cup \First(\E_2))\cap{\Sigma}_0 & \mbox{ if } c \in\llbracket{\E_1}\rrbracket,\\
             \ \First(\E_1)\cap{\Sigma}_0 & \text{ otherwise.}\\ 
          \end{array}
        \right.\\
&=\left\{
    \begin{array}{l@{\ }l}
((\firstt{\E_1}\uplus \firs{\E_1}\setminus\{c\}))\cup (\firstt{\E_2}\uplus\firs{\E_2})~)\cap{\Sigma}_0 & \mbox{ if } c \in\llbracket{\E_1}\rrbracket,\\
     (\firstt{\E_1}\uplus \firs{\E_1})\cap{\Sigma}_0 & \text{ otherwise.}\\ 
          \end{array}
        \right.\\
&=\left\{
    \begin{array}{l@{\ }l}
(\firs{\E_1}\setminus\{c\})\cup \firs{\E_2}~) & \mbox{ if } c \in\llbracket{\E_1}\rrbracket,\\
    \ \firs{\E_2} & \text{ otherwise.}\\ 
          \end{array}
        \right.\\
\end{align*}
\end{enumerate}
\end{sloppy}  
  \qed
\end{proof} 
The following proposition shows that $\firstt{\E}$ can be computed in a similar way to the case of words. 
\begin{proposition}
\label{prop firstsup}
Let $\E$ be a linear regular expression and $\h$ be a subexpression of $\E$. The set of symbols $\firstt{\h}$ is defined as: 
\begin{eqnarray*}
\firstt{a}=\firstt{0}&=& \emptyset,\\
\firstt{f_j(\E_1, \cdots,\E_m)}&=&\lbrace f_j\rbrace,\\
\firstt{\E_1+\E_2}&=&\firstt{\E_1}\uplus\firstt{\E_2},\\
\firstt{{\E_1}\cdot_c \E_2}&=&\left\{
\begin{array}{ll}
\firstt{\E_1}\uplus\firstt{\E_2}&\mbox{ if } c\in\llbracket\E_1\rrbracket,\\
\firstt{\E_1}&\mbox{ otherwise.}
\end{array}\right.\\
\firstt{{\E_1}^{*_c}}&=&\firstt{\E_1},\\
\end{eqnarray*}
\end{proposition}
\begin{proof}
\begin{sloppy}  
  
Let $\E$ be a linear regular expression. 
\begin{enumerate}
\item If $\E=0$ or if $\E=f_j(\E_1,\ldots,\E_m)$, then $\firs{\E}=\emptyset$ and for $\E=a$, $\firs{\E}=\{a\}$.  

\item Let us prove this proposition for the cases $\E=\E_1\cdot_c \E_2$. 

We have $\firs{\E_1\cdot_c \E_2}=\First(\E_1\cdot_c \E_2,f_j,k)\cap{\Sigma}_{>} $
 \begin{align*} 
\firstt{\E_1\cdot_c \E_2}&=\First(\E_1\cdot_c \E_2)\cap{\Sigma}_{>} \\ 
&=\left\{
          \begin{array}{l@{\ }l}
            (\First(\E_1)\setminus\{c\}\cup \First(\E_2))\cap{\Sigma}_{>}  & \mbox{ if } c \in\llbracket{\E_1}\rrbracket,\\
             \ \First(\E_1)\cap\Sigma_{>} & \text{ otherwise.}\\ 
          \end{array}
        \right.\\
&=\left\{
    \begin{array}{l@{\ }l}
(~((\firstt{\E_1}\uplus \firs{\E_1})\setminus\{c\})\cup (\firstt{\E_2}\uplus\firs{\E_2})~)\cap{\Sigma}_{>}  & \mbox{ if } c \in\llbracket{\E_1}\rrbracket,\\
     (\firstt{\E_1}\uplus \firs{\E_1})\cap{\Sigma}_{>}  & \text{ otherwise.}\\ 
          \end{array}
        \right.\\
        &=\left\{
    \begin{array}{l@{\ }l}
(~((\firs{\E_1}\setminus\{c\})\uplus \firstt{\E_1})\cup (\firstt{\E_2}\uplus\firs{\E_2})~)\cap{\Sigma}_{>}  & \mbox{ if } c \in\llbracket{\E_1}\rrbracket,\\
     (\firstt{\E_1}\uplus \firs{\E_1})\cap{\Sigma}_{>} & \text{ otherwise.}\\ 
          \end{array}
        \right.\\
&=\left\{
    \begin{array}{l@{\ }l}
\firstt{\E_1}\cup \firstt{\E_2} & \mbox{ if } c \in\llbracket{\E_1}\rrbracket,\\
    \ \firstt{\E_2} & \text{ otherwise.}\\ 
          \end{array}
        \right.\\
\end{align*}
\end{enumerate}
\end{sloppy}  
\qed
\end{proof}

Let us recall that $\las{\E}{f_j}{k}$ and $\Fw{\E}{f_j}{k}$ are, respectively, the projection of the set $\Fll{\E}{f_j}{k}$ to the symbols associated with symbols of a rank $0$ and a rank greater than $0$. We have: 
\begin{eqnarray*}
\Fll{\E}{f_j}{k}=\las{\E}{f_j}{k}\uplus\Fw{\E}{f_j}{k}
\end{eqnarray*}
\begin{proposition}\label{a}
Let $\E$ be a linear regular expression, $1\leq k\leq m$ be two integers and $f_j$ be a symbol in $\Sigma^{\E}_m$. 
The function $\las{\E}{f_j}{k}$ can be computed inductively as follows: 
\begin{eqnarray*}
\las{a}{f_j}{k}&=& \las{0}{f_j}{k}=\emptyset,\\
\las{g_i(\E_1, \cdots,\E_m)}{f_j}{k}&=&\left\{
\begin{array}{lll}
\displaystyle\firs{\E_k}& &\;\;\;\;\;\;\;\;\;\;\;\;\mbox{ if } g_i=f_j, \\ 
\las{\E_l}{f_j}{k}& &\;\;\;\;\;\;\;\;\;\;\;\;\mbox{ if } f_j\in \Sigma^{\E_l},
\end{array}\right.\\
\las{\E_1+\E_1}{f_j}{k}&=&\left\{
\begin{array}{ll}
\las{\E_1}{f_j}{k} &\;\;\;\;\;\;\;\;\;\;\;\;\;\;\;\;\mbox{ if }f_j\in \Sigma^{\E_1}, \\ 
\las{\E_2}{f_j}{k} &\;\;\;\;\;\;\;\;\;\;\;\;\;\;\;\;\mbox{ if } f_j\in \Sigma^{\E_2},
\end{array}\right.\\
\las{\E_1\cdot_c\E_1}{f_j}{k}&=&\left\{
\begin{array}{ll}
(\las{\E_1}{f_j}{k}\setminus\{c\})\cup\firs{\E_2}& \mbox{ if }f_j\in \Sigma^{\E_1}\\ 
&\mbox{ and } c\in \las{\E_1}{f_j}{k},\\ 
\las{\E_1}{f_j}{k}&\mbox{ if } f_j\in \Sigma^{\E_1}\\
&\mbox{ and }  c\notin \las{\E_1}{f_j}{k},\\
\las{\E_2}{f_j}{k}&\mbox{ if } f_j\in \Sigma^{\E_2}\\
&\mbox{ and } c\in  \Last(\E_1),\\
\emptyset&\mbox{ otherwise.}
\end{array}\right.\\
 \las{{\E_1}^{*_c}}{f_j}{k}&=&\left\{
\begin{array}{lll}
(\las{\E_1}{f_j}{k}\setminus\{c\})\cup\firs{{\E_1}}&\mbox{ if } c\in \las{\E_1}{f_j}{k},\\
\las{\E_1}{f_j}{k} &\mbox{ otherwise.}
\end{array}\right.
\end{eqnarray*}
\end{proposition}

\begin{proof}
\begin{sloppy}  
  
Let $\E$ be a linear regular expression, $1\leq k\leq m$ be two integers and $f_j$ be a symbol in $\Sigma^{\E}_m$.  
\begin{enumerate}
\item If $\E=0$ or if $\E=a$, then $\las{\E}{f_j}{k}=\emptyset$.

Let us prove this proposition for the cases $\E=\E_1\cdot_c \E_2$ and $\E=\E_1^{*_c}$. 
\item Let us consider that $\E=\E_1\cdot_c \E_2$. 

We have $\las{\E_1\cdot_c \E_2}{f_j}{k}=\Follow(\E_1\cdot_c \E_2,f_j,k)\cap{\Sigma}_0$
\begin{align*} 
\las{\E_1\cdot_c \E_2}{f_j}{k}&=\Follow(\E_1\cdot_c \E_2,f_j,k)\cap{\Sigma}_0\\ 
&=\left\{
\begin{array}{l@{\ }l}
((\Follow(\E_1,f_j,k)\setminus\{c\})\cup\First(\E_2))\cap{\Sigma}_0 &\text{ if } f_j\in\Sigma^{\E_1}\wedge c\in\las{\E_1}{f_j}{k},\\
            \Follow(\E_1,f_j,k)\cap{\Sigma}_0 & \text{ if } f_j\in \Sigma^{\E_1}\wedge c\notin \las{\E_1}{f_j}{k},\\
            \Follow(\E_2,f_j,k)\cap{\Sigma}_0 & \text{ if } f_j\in \Sigma^{\E_2}\wedge c\in\Last(\E_1),\\
            \emptyset & \text{ otherwise.}
\end{array}
\right.\\ 
          &=\left\{
\begin{array}{l@{\ }l}
(((\Fw{\E_1}{f_j}{k}\uplus \las{\E_1}{f_j}{k})\setminus\{c\})\cup &\ \text{ if }f_j\in \Sigma^{\E_1} \wedge c\in\las{\E_1}{f_j}{k},\\ \ (\firs{\E_2}\uplus\firstt{\E_2})~)\cap{\Sigma}_0             & \\
 (\Fw{\E_1}{f_j}{k}\uplus \las{\E_1}{f_j}{k})\cap{\Sigma}_0 & \text{ if } f_j\in \Sigma^{\E_1} \wedge c\notin\las{\E_1}{f_j}{k},\\
 (\Fw{\E_2}{f_j}{k}\uplus \las{\E_2}{f_j}{k})\cap{\Sigma}_0 & \text{ if } f_j\in \Sigma^{\E_2}  \wedge c\in\Last(\E_1),\\
 \emptyset & \text{ otherwise.}
\end{array}
\right.\\ 
&=\left\{
\begin{array}{l@{\ }l}
 \las{\E_1}{f_j}{k}\cup \firs{\E_2} & \mbox{ if } f_j\in \Sigma^{\E_1} \mbox{ and }c\in\las{\E_1}{f_j}{k},\\
  \las{\E_1}{f_j}{k} & \mbox{ if } f_j\in \Sigma^{\E_1} \mbox{ and }c\notin\las{\E_1}{f_j}{k},\\ 
  \las{\E_2}{f_j}{k} & \mbox{ if } f_j\in \Sigma^{\E_2} \mbox{ and }c\in\Last(\E_1),\\ 
  \emptyset & \mbox{ otherwise.}\\ 
\end{array}
\right.
\end{align*}
\item Let us consider that $\E=\E_1^{*_c}$. By definition we have $\las{\E_1^{*_c}}{f_j}{k}=\Follow(\E_1^{*_c},f_j,k)\cap {\Sigma}_0$. Then: 
 \begin{align*} 
\las{\E_1^{*_c}}{f_j}{k}&=\Follow(\E_1^{*_c},f_j,k)\cap{\Sigma}_0\\     
&=\left\{
          \begin{array}{l@{\ }l}
((\Follow(\E_1,f_j,k)\setminus\{c\})\cup \First(\E_1))\cap{\Sigma}_0 & \text{ if }c\in\las{\E_1}{f_j}{k},\\
             \ (\Follow(\E_1,f_j,k)\cap{\Sigma}_0 & \text{ otherwise.}\\ 
          \end{array}
        \right.\\
        &=\left\{
          \begin{array}{l@{\ }l}
            (~(\Fw{\E_1}{f_j}{k}\uplus \las{\E_1}{f_j}{k})\setminus\{c\})\cup & \text{ if }c\in\las{\E_1}{f_j}{k},\\
            \  \ \ (\firs{\E_1}\uplus\firstt{\E_1})~)\cap{\Sigma}_0             & \\
             \ (\Fw{\E_1}{f_j}{k}\uplus \las{\E_1}{f_j}{k})\cap{\Sigma}_0 & \text{ otherwise.}\\ 
          \end{array}
        \right.\\
         &=\left\{
          \begin{array}{l@{\ }l}
  \las{\E_1}{f_j}{k}\cup \firs{\E_1} & \mbox{ if } c\in\las{\E_1}{f_j}{k}),\\
   \las{\E_1}{f_j}{k} & \mbox{ otherwise.}\\ 
          \end{array}
        \right.
\end{align*}
\end{enumerate}
\end{sloppy}   
\qed
\end{proof}

\begin{proposition}
\label{x}
Let $\E$ be a linear regular expression, $1\leq k\leq m$ be two integers and $f_j$ be a symbol in $\Sigma^{\E}_m$.  
We define inductively the set $\Fw{\E}{f_j}{k}$ as follows:
\begin{eqnarray*}
\Fw{a}{f_j}{k}&=&\Fw{0}{f_j}{k}=\emptyset, \\
 \Fw{g_i(\E_1, \dots,\E_m)}{f_j}{k}&=&
 \left\{
\begin{array}{ll}
\displaystyle\firstt{\E_k}&\;\;\;\;\;\;\;\;\;\;\;\;\;\;\;\mbox{ if } g_i=f_j,\\ 
\Fw{\E_l}{f_j}{k}&\;\;\;\;\;\;\;\;\;\;\;\;\;\;\;\mbox{ if } f_j\in \Sigma^{\E_l}.
\end{array}\right.\\
\Fw{\f+\G}{f_j}{k}&=&\left\{
\begin{array}{ll}
\Fw{\f}{f_j}{k} &\;\;\;\;\;\;\;\;\;\;\;\;\;\;\;\;\mbox{ if }f_j\in \Sigma^{\f}, \\ 
\Fw{\G}{f_j}{k}&\;\;\;\;\;\;\;\;\;\;\;\;\;\;\;\;\mbox{ if } f_j\in \Sigma^{\G},
\end{array}\right.\\
\Fw{\f\cdot_c \G}{f_j}{k}&=&\left\{
\begin{array}{ll}
\Fw{\f}{f_j}{k}\cup\firstt{\G}&\mbox{ if }c\in \las{\f}{f_j}{k},\\
\Fw{\f}{f_j}{k}&\mbox{ if } f_j\in \Sigma^{\f},\\
&\mbox{ and } c\notin \las{\f}{f_j}{k}, \\
\Fw{\G}{f_j}{k}&\mbox{ if } f_j\in \Sigma^{\G}\\
&\mbox{ and } c\in  \Last(\f),\\
\emptyset&\mbox{ otherwise.}
\end{array}\right.\\
\Fw{{\f}^{*_c}}{f_j}{k}&=&\left\{
\begin{array}{ll}
\Fw{\f}{f_j}{k}\cup\firstt{\f}&\mbox{ if }c\in \las{\f}{f_j}{k},\\
\Fw{\f}{f_j}{k}&\mbox{ otherwise.}
\end{array}\right.\\
\end{eqnarray*}
\end{proposition}
\begin{proof}
\begin{sloppy}  
  
Let $\E$ be a linear regular expression, $1\leq k\leq m$ be two integers and $f_j$ be a symbol in $\Sigma^{\E}_m$. 
\begin{enumerate}
\item If $\E=0$ or if $\E=a$, then $\Fw{\E}{f_j}{k}=\emptyset$.\\
Let us prove this proposition for the cases $\E=\E_1\cdot_c \E_2$ and $\E=\E_1^{*_c}$. 
\item Let us consider that $\E=\E_1\cdot_c \E_2$. 

We have $\Fw{\E_1\cdot_c \E_2}{f_j}{k}=\Follow(\E_1\cdot_c \E_2,f_j,k)\cap\Sigma_{>}$
\begin{align*} 
\Fw{\E_1\cdot_c \E_2}{f_j}{k}&=\Follow(\E_1\cdot_c \E_2,f_j,k)\cap\Sigma_{>}\\ 
&=\left\{
		  \begin{array}{l@{\ }l}
            (~(\Follow(\E_1,f_j,k)\setminus\{c\}) \cup \First(\E_2)~)\cap\Sigma_{>} & \text{ if }  f_j\in \Sigma^{\E_1} \ \wedge \\
             &\ \ c\in\las{\E_1}{f_j}{k},\\
            \Follow(\E_1,f_j,k)\cap\Sigma_{>} & \text{ if } f_j\in \Sigma^{\E_1}  \  \wedge \\
            &\ c\notin\las{\E_1}{f_j}{k},\\
            \Follow(\E_2,f_j,k)\cap\Sigma_{>} & \text{ if } f_j\in \Sigma^{\E_2} \  \wedge c\in\Last(\E_1),\\
            \emptyset & \text{ otherwise.}
          \end{array}
          \right.\\ 
          &=\left\{
		  \begin{array}{l@{\ }l}
(~((\Fw{\E_1}{f_j}{k}\uplus \las{\E_1}{f_j}{k})\setminus\{c\})\cup&\text{ if }f_j\in\Sigma^{\E_1}\wedge c\in\las{\E_1}{f_j}{k},\\     
\ (\firs{\E_2}\uplus\firstt{\E_2})~)\cap\Sigma_{>} \\	& \\
(\Fw{\E_1}{f_j}{k}\uplus \las{\E_1}{f_j}{k})\cap\Sigma_{>} & \text{ if } f_j\in \Sigma^{\E_1}\ \wedge c\notin \las{\E_1}{f_j}{k},\\
            & \\
            (\Fw{\E_2}{f_j}{k}\uplus \las{\E_2}{f_j}{k})\cap\Sigma_{>} & \text{ if } f_j\in \Sigma^{\E_2}  \wedge c\in\Last(\E_1),\\
            \emptyset & \text{ otherwise.}
          \end{array}
          \right.\\ 
 &=\left\{
          \begin{array}{l@{\ }l}
  \Fw{\E_1}{f_j}{k}\cup \firstt{\E_2} & \mbox{ if } f_j \in\Sigma^{\E_1} \mbox{ and }c\in\las{\E_1}{f_j}{k}),\\
  \Fw{\E_1}{f_j}{k} & \mbox{ if } f_j \in\Sigma^{\E_1} \mbox{ and }c\notin\las{\E_1}{f_j}{k}),\\ 
  \Fw{\E_2}{f_j}{k} & \mbox{ if } f_j \in\Sigma^{\E_2} \mbox{ and }c\in\Last(\E_1),\\ 
  \emptyset & \mbox{ otherwise.}\\ 
          \end{array}
        \right.
\end{align*}
\item Let us consider that $\E=\E_1^{*_c}$. By definition we have $\Fw{\E_1^{*_c}}{f_j}{k}=\Follow(\E_1^{*_c},f_j,k)\cap \Sigma_{>}$. Then: 
 \begin{align*} 
\Fw{\E_1^{*_c}}{f_j}{k}&=\Follow(\E_1^{*_c},f_j,k)\cap\Sigma_{>}\\       
&=\left\{
          \begin{array}{l@{\ }l}
((\Follow(\E_1,f_j,k)\setminus\{c\})\cup \First(\E_1))\cap\Sigma_{>} & \text{ if } c\in\las{\E_1}{f_j}{k},\\
             \ (\Follow(\E_1,f_j,k)\cap\Sigma_{>} & \text{ otherwise.}\\ 
          \end{array}
        \right.\\
        &=\left\{
          \begin{array}{l@{\ }l}
            ((\Fw{\E_1}{f_j}{k}\uplus \las{\E_1}{f_j}{k})\setminus\{c\})\cup & \text{ if } c\in\las{\E_1}{f_j}{k},\\
            \  \ \ (\firs{\E_1}\uplus\firstt{\E_1}))\cap \Sigma_{>} \\            & \\
             \ (\Fw{\E_1}{f_j}{k}\uplus \las{\E_1}{f_j}{k})\cap\Sigma_{>} & \text{ otherwise.}\\ 
          \end{array}
        \right.\\
         &=\left\{
          \begin{array}{l@{\ }l}
   \Fw{\E_1}{f_j}{k}\cup \firstt{\E_1} & \mbox{ if } c\in\las{\E_1}{f_j}{k}),\\
   \Fw{\E_1}{f_j}{k} & \mbox{ otherwise.}\\ 
          \end{array}
        \right.
\end{align*}
\end{enumerate}
\end{sloppy}    
\qed
\end{proof}
\begin{remark}\label{remark}
The definition of the set $\Fw{\E}{f_j}{k}$ is identical to the function $\Follow$ in the case of words~\cite{ZPC}. We have the same formulas.
\end{remark}

The construction of the $k$-position tree automaton ${\cal P}_{\E}$ from the regular expression  as it has been presented in this article complies with the properties of the position automaton proposed by Glushkov. 
This is the generalization of the position automaton from words to trees.

\subsection{$\mathrm{ZPC}$-Structure for $\Follow$ Computation}
In the word case, the construction of the position automaton, has been developed in~\cite{ZPC96,ZPC}. This construction will be extended to trees in the following.
 
Let $T_{\E}$ be the syntax tree associated with the regular expression  $\E$. 

The set of nodes of $T_{\E}$ is written as $\mathrm{Nodes}(\E)$. For a node $\nu$ in $\mathrm{Nodes}(\E)$, $\mathrm{sym}(\nu)$, $\mathrm{father}(\nu)$, $\mathrm{son}(\nu)$, $\mathrm{right}(\nu)$ and $\mathrm{left}(\nu)$ denote respectively the symbol, the father, the son, the right son and the left son of the node $\nu$ if they exist. 

We denote by $\E_{\nu}$ the subexpression rooted at $\nu$; In this case we write $\nu_{\E}$ to denote the node associated to $\E_{\nu}$. Let $\gamma:~\mathrm{Nodes}(\E)\cup\{\bot\}\rightarrow ~\mathrm{Nodes}(\E)\cup\{\bot\}$ be the function defined by:

    $$\gamma(\nu)= \left\{
    \begin{array}{l@{\ }l}
    \mathrm{father}(\nu) & \mbox{ if } \mathrm{sym}(\mathrm{father}(\nu))=^{*_c} \mbox{ and }\nu\neq \nu_{\E}\\
   \mathrm{right}(\mathrm{father}(\nu)) & \mbox{ if } \mathrm{sym}(\mathrm{father}(\nu))=\cdot_c\\
    \bot & \mbox{ otherwise }\\
    \end{array}\right.$$

\noindent where $\bot$ is an artificial node such that $\gamma(\bot)=\bot$. The $\mathrm{ZPC}$-Structure is the syntax tree equipped with $\gamma(\nu)$ links.

We extend the relation $\preccurlyeq$ to the set of nodes of $T_{\E}$: For two nodes $\mu$ and $\nu$ we write $\nu\preccurlyeq\mu\Leftrightarrow T_{\E_{\nu}}\preccurlyeq T_{\E_{\mu}}$. We define the set $\Gamma_{\nu}(\E)=\{\mu\in\mathrm{Nodes}(\E)\mid \nu \preccurlyeq \mu \land \gamma(\mu)\neq \bot\}$ which is totally ordered by $\preccurlyeq$.

\begin{proposition}\label{b}
Let $\E$ be linear regular expression, $1\leq k\leq n$ be two integers and $f$ be in $\Sigma^{\E}\cap\Sigma_n$. Then $\Follow(\E,f,k)= ((((\First(\E_{\nu_0})\cdot_{op(\nu_1)} \First(\E_{\gamma(\nu_1)}))\cdot_{op(\nu_2)} \First(\E_{\gamma(\nu_2)}))
\dots \cdot_{op(\nu_{m})} \First(\E_{\gamma(\nu_m)}))$ where $\nu_f$ is the node of $T_{\E}$ labelled by $f$, $\nu_0$ is the $k\mbox{-}\mathrm{child}(\nu_{f})$, $\Gamma_{\nu_{f}}(\E)=\{\nu_1, \dots,\nu_m \}$ and for $1\leq i\leq m,~op(\nu_i)=c$ such that $\mathrm{sym}(\mathrm{father}(\nu_i)) \in \{\cdot_c,{*_c}\}$.
\end{proposition}
\begin{proof}
\begin{sloppy}
  By induction over the structure of $E$.
  \begin{enumerate}
    \item Let us suppose that $E=f(\E_1,\ldots,\E_n)$. Then $\Follow(\E,f,k)=\First(\E_k)$. Since by definition $\nu_f$ is the root of $T_{\E}$, $k\mbox{-}\mathrm{child}(\nu_{f})$ is the root of $\E_{\nu_0}=E_k$. Hence $\First(\E_{\nu_0})=\First(\E_k)= \Follow(\E,f,k)$.
    \item Let us suppose that $E=g(\E_1,\ldots,\E_m)$ with $g\neq f$, or $\E=\E_1+\E_2$, or $\E=\E_1 \cdot_c \E_2$ with $f\in \Sigma^{\E_2}$. Then $\Follow(\E,f,k)=\Follow(\E_j,f,k)$ with $f\in \Sigma^{\E_j}$. By induction hypothesis, $\Follow(\E_j,f,k)=((((\First(\E_{\nu_0}) \cdot_{op(\nu_1)} \First(\E_{\gamma(\nu_1)}))\cdot_{op(\nu_2)} \First(\E_{\gamma(\nu_2)})) \dots \cdot_{op(\nu_{m})} \First(\E_{\gamma(\nu_m)}))$ where $\nu_f$ is the node of $T_{\E_j}$ labelled by $f$, $\nu_0$ is the $k\mbox{-}\mathrm{child}(\nu_{f})$, $\Gamma_{\nu_{f}}(\E_j)=\{\nu_1, \dots,\nu_m \}$ and for $1\leq i\leq m,op(\nu_i)=c$ such that $\mathrm{sym}(\mathrm{father}(\nu_i)) \in \{\cdot_c,{*_c}\}$. Since $T_{\E_j} \preccurlyeq T_{\E}$, $\Follow(\E_j,f,k)=((((\First(\E_{\nu_0}) \cdot_{op(\nu_1)} \First(\E_{\gamma(\nu_1)}))\cdot_{op(\nu_2)} \First(\E_{\gamma(\nu_2)})) \dots \cdot_{op(\nu_{m})} \First(\E_{\gamma(\nu_m)}))$ where $\nu_f$ is the node of $T_{\E}$ labelled by $f$, $\nu_0$ is the $k\mbox{-}\mathrm{child}(\nu_{f})$, $\Gamma_{\nu_{f}}(\E_j)=\{\nu_1, \dots,\nu_m \}$ and for $1\leq i\leq m,op(\nu_i)=c$ such that 
$\mathrm{sym}(\mathrm{father}(\nu_i)) \in \{\cdot_c,{*_c}\}$.    
    \item Let us suppose that $\E=\E_1 \cdot_c \E_2$ with $f\in\Sigma^{\E_1} $ (resp. $\E=\E_1^{*_c}$). Then $\Follow(\E,f,k)=\Follow(\E_1,f,k)\cdot_c \First(\G )$ with $G\in\{\E_1^{*_c},\E_2\}$. By induction hypothesis, $\Follow(\E_1,f,k)= ((((\First(\E_{\nu_0})\cdot_{op(\nu_1)} \First(\E_{\gamma(\nu_1)}))\cdot_{op(\nu_2)} \First(\E_{\gamma(\nu_2)})) \dots \cdot_{op(\nu_{m})} \First(\E_{\gamma(\nu_m)}))$ where $\nu_f$ is the node of $T_{\E_1}$ labelled by $f$, $\nu_0$ is the $k\mbox{-}\mathrm{child}(\nu_{f})$, $\Gamma_{\nu_{f}}(\E_j)=\{\nu_1, \dots,\nu_m \}$ and for $1\leq i\leq m,op(\nu_i)=c$ such that $\mathrm{sym}(\mathrm{father}(\nu_i)) \in \{\cdot_c,{*_c}\}$. 
    
    Since $T_{\E_1} \preccurlyeq T_{\E}$, by setting $\h=\E_{\nu_{m+1}}$ and $op(\nu_{m+1})=c$, $\Follow(\E_1,f,k)\cdot_c \First(\h)=((((\First(\E_{\nu_0}) \cdot_{op(\nu_1)} \First(\E_{\gamma(\nu_1)})) \cdot_{op(\nu_2)} \First(\E_{\gamma(\nu_2)})) \dots \cdot_{op(\nu_{m})} \First(\E_{\gamma(\nu_m)})) \cdot_{op(\nu_{m+1})} \First(\E_{\gamma(\nu_{m+1})})$ where $\nu_f$ is the node of $T_{\E}$ labelled by $f$, $\nu_0$ is the $k\mbox{-}\mathrm{child}(\nu_{f})$, $\Gamma_{\nu_{f}}(\E)=\{\nu_1, \dots,\nu_m,\nu_{m+1}\}$ and for $1\leq i\leq m+1,op(\nu_i)=c$ such that $\mathrm{sym}(\mathrm{father}(\nu_i)) \in \{\cdot_c,{*_c}\}$.
  \end{enumerate}
\end{sloppy}  
  \qed
\end{proof}

\subsection{Description of the algorithm and complexity}

An implicit construction of the word position automaton, the so-called ZPC-structure, has been developed by Ziadi et al. \cite{ZPC96,ZPC}.
Algorithm~\ref{zpc} extends this construction to the regular tree expressions. It constructs
a forest of trees where every tree rooted at a node $\nu_{\f}$ represents the set $\firstt{\f}$ according to Proposition~\ref{prop firstsup}. 

\begin{algorithm}[H]
\caption{$\mathrm{ZPC}$-Structure Construction}
\label{zpc}
\KwIn{Regular Expression $\E$.}
\KwOut{$\mathrm{ZPC}$-Structure}
 Construct the syntax tree  $T_{\E}$ of $\E$;\\
{\#}\\
\For{each node $\nu_{\f}$ on $T_{\E}$}{
	 Compute $\firs{\f}$\;
\textbf{end for}  }
{\# The construction of a $\First$ Forest}\\
\For{each node $\nu_{\f\cdot_c \G}$ in $T_{\E}$}{
	\If{$c\notin \firs{\f}$}{
 		 Remove the link $(\nu_{\f\cdot_c \G}, \nu_{\G})$\;
	\textbf{end if}}
\textbf{end for}}
{\# We have $\first{f_j(\E_1,\ldots,\E_n)}=\{f_j\}$}\\
\For{each node $\nu_{f_j(\E_1,\ldots\E_n)}$ in $T_{\E}$}{
	\For{$i=1$ to $n$}{
 		 Remove the link $(\nu_{f_j(\E_1,\ldots\E_n)}, \nu_{\E_i})$\;
		\textbf{end for}}
\textbf{end for}}
\For{each node $\nu_{\f}\in \Sigma_0$}{
Delete the node $\nu_{\f}$\; 
\textbf{end for}}
{\# }\\
{\# The construction of $\Follow$ links ($\gamma_{\nu}$ links)}\\
 \For{each node $\nu_{\f \cdot_c \G}$ in $T_{\E}$}{
  create a follow link from $\nu_{\f}$ to $\nu_{\G}$\;
  \textbf{end for}}
 \For{each node $\nu_{\f^{*_c}}$ in $T_{\E}$}{
  create a link from $\nu_{\f}$ to $\nu_{{\f}^{*_c}}$\;
  \textbf{end for} }
  {\ }\\
     \Return{$\mathrm{ZPC}$-Structure}
\end{algorithm}
\begin{example}
The syntax tree  $T_{\E}$ associated with the regular expression  $\E=(f_1(a)^{*_a}\cdot_a b+ h_2(b))^{*_b}+g_3(c,a)^{*_c}\cdot_c (f_4(a)^{*_a}\cdot_a b+ h_5(b))^{*_b}$ is given in Figure~\ref{fig a t e}. 
	\end{example}

\begin{figure}[H] 
	 \centerline{
	\begin{tikzpicture}[node distance=1.3cm,bend angle=30,transform shape,scale=.75]
		\node(1) {$+$};		
		\node[below right of=1,xshift=1.5cm] (3) {$\cdot_c$};
		\node[below left of=1,xshift=-1.5cm] (22) {$*_b$};			
		\node[below of=22,node distance=1cm] (221) {$+$};		
		\node[below left of=221,xshift=-.75cm] (5) {$+$};	 
		\node[below right of=221,xshift=.75cm] (222) {$b$};  
		\node[below right of=5] (16) {$h_2$};		 
		\node[below of=16] (4) {$b$};
		\node[below left of=5] (17) {$\cdot_a$};
		\node[below left of=17] (18) {$*_a$}; 
	 \node[below right of=17] (19) {$b$};
	  \node[below of=18] (181) {$+$};
	   \node[below of=181] (182) {$f_1$};
	   \node[right of=182] (184) {$a$};
	   \node[below of=182] (183) {$a$};
		\node[below right of=3,xshift=.75cm] (6) {$*_b$};
		\node[below of=6] (61) {$+$};
        \node[below right of=61] (62) {$b$}; 
		\node[below of=61] (51) {$+$};
		\node[below right of=51] (116) {$h_5$};		 
		\node[below of=116] (41) {$b$};
		\node[below left of=51] (117) {$\cdot_a$};
		\node[below left of=117] (118) {$*_a$}; 
	 \node[below right of=117] (119) {$b$};
	  \node[node distance=1cm,below of=118] (1181) {$+$};
	   \node[node distance=1cm,below of=1181] (1182) {$f_4$};
	  \node[node distance=1cm,below of=1182] (1183) {$a$};
	  \node[node distance=1cm,right of=1182] (1184) {$a$};
	\node[below left of=3,xshift=-.75cm] (31) {$*_c$};	
	\node[below of=31] (3121) {$+$};
	\node[below of=3121] (312) {$g_3$};
	\node[below right of=3121] (3122) {$c$};	 
	\node[below right of=312] (311) {$a$};
	\node[below left of=312] (313) {$c$};
	\node[inner sep=0pt,draw, circle, fit=(1183), color=blue] (eqa1) {};
	\node[inner sep=0pt,draw, circle, fit=(1184), color=blue] (eqa1) {};
	   \node[inner sep=0pt,draw, circle, fit=(51), color=blue] (eqa1) {};
	   \node[inner sep=0pt,draw, circle, fit=(116), color=blue] (eqa1) {};
	   \node[inner sep=0pt,draw, circle, fit=(41), color=blue] (eqa1) {};
	   \node[inner sep=0pt,draw, circle, fit=(117), color=blue] (eqa1) {};
	   \node[inner sep=0pt,draw, circle, fit=(118), color=blue] (eqa1) {};	   
	   \node[inner sep=0pt,draw, circle, fit=(119), color=blue] (eqa1) {};
	   \node[inner sep=0pt,draw, circle, fit=(1181), color=blue] (eqa1) {};
	   \node[inner sep=0pt,draw, circle, fit=(1182), color=blue] (eqa1) {};
	 \node[inner sep=0pt,draw, circle, fit=(1), color=blue] (eq1) {};
\node[inner sep=0pt,draw, circle, fit=(181), color=blue] (eqf1) {};
\node[inner sep=0pt,draw, circle, fit=(182), color=blue] (eqa1) {};		
\node[inner sep=0pt,draw, circle, fit=(183), color=blue] (eqa1) {};	
\node[inner sep=0pt,draw, circle, fit=(184), color=blue] (eqa1) {};
	    \node[inner sep=0pt,draw, circle, fit=(22), color=blue] (eq2) {};	    
	    \node[inner sep=0pt,draw, circle, fit=(221), color=blue] (eq2) {};   
	    \node[inner sep=0pt,draw, circle, fit=(222), color=blue] (eq2) {};
	    \node[inner sep=0pt,draw, circle, fit=(3), color=blue] (eq3) {};
	    \node[inner sep=0pt,draw, circle, fit=(5), color=blue] (eq4) {};
	    \node[inner sep=0pt,draw, circle, fit=(4), color=blue] (eq5) {};
	    \node[inner sep=0pt,draw, circle, fit=(16), color=blue] (eq6) {};
	    \node[inner sep=0pt,draw, circle, fit=(19), color=blue] (eq7) {};
	    \node[inner sep=0pt,draw, circle, fit=(18), color=blue] (eq8) {};	
	    \node[inner sep=0pt,draw, circle, fit=(17), color=blue] (eq9) {};
	    \node[inner sep=0pt,draw, circle, fit=(62), color=blue] (eq10) {};  
	    \node[inner sep=0pt,draw, circle, fit=(61), color=blue] (eq10) {};    
	    \node[inner sep=0pt,draw, circle, fit=(6), color=blue] (eq10) {};
	   \node[inner sep=0pt,draw, circle, fit=(311), color=blue] (eq11) {};
	    \node[inner sep=0pt,draw, circle, fit=(313), color=blue] (eq121) {};
	   \node[inner sep=0pt,draw, circle, fit=(3121), color=blue] (eq12) {};
	   \node[inner sep=0pt,draw, circle, fit=(312), color=blue] (eq12) {};
	   \node[inner sep=0pt,draw, circle, fit=(3122), color=blue] (eq12) {};
	   \node[inner sep=0pt,draw, circle, fit=(31), color=blue] (eq13) {};
	 \path[->]
	 (18) edge node[above left,pos=0.4] {} (181)
	 (181) edge node[above left,pos=0.4] {} (182)
	 (182) edge node[above left,pos=0.4] {} (183)
	 (181) edge node[above left,pos=0.4] {} (184)
	      (1) edge node[above left,pos=0.4] {} (22)
	      (1) edge node[above right,pos=0.4] {} (3)
	      (16) edge node[left,pos=0.4] {} (4)
	      (61) edge node[left,pos=0.4] {} (62)
	      (6) edge node[left,pos=0.4] {} (61)
	      (61) edge node[left,pos=0.4] {} (51)
	      (22) edge node[left,pos=0.4] {} (221)
	      (221) edge node[left,pos=0.4] {} (222)
	      (221) edge node[left,pos=0.4] {} (5)
	      (3) edge node[right,pos=0.4] {} (6)
	      (5) edge node[above left,pos=0.4] {} (16)
	      (5) edge node[above left,pos=0.4] {} (17)
	      (17) edge node[left,pos=0.4] {} (18)
	      (17) edge node[left,pos=0.4] {} (19)
	      (3) edge node[left,pos=0.4] {} (31)
	      (31) edge node[left,pos=0.4] {} (3121)
	      (3121) edge node[left,pos=0.4] {} (3122)
	      (3121) edge node[left,pos=0.4] {} (312)
	      (312) edge node[left,pos=0.4] {} (311)
	      (312) edge node[left,pos=0.4] {} (313)
	      (51) edge node[left,pos=0.4] {} (116)
	      (116) edge node[left,pos=0.4] {} (41)
	      (51) edge node[left,pos=0.4] {} (117)
	      (117)edge node[left,pos=0.4] {} (118)
	      (117)edge node[left,pos=0.4] {} (119)
	      (118)edge node[left,pos=0.4] {} (1181)(1181)edge node[left,pos=0.4] {} (1184)
	      (1181)edge node[left,pos=0.4] {} (1182)(1182)edge node[left,pos=0.4] {} (1183)
	     ;
	     \end{tikzpicture}
	 }
\caption{The syntax tree  $T_{\b\E}$ of $\b\E$}
	 \label{fig a t e}
	\end{figure}
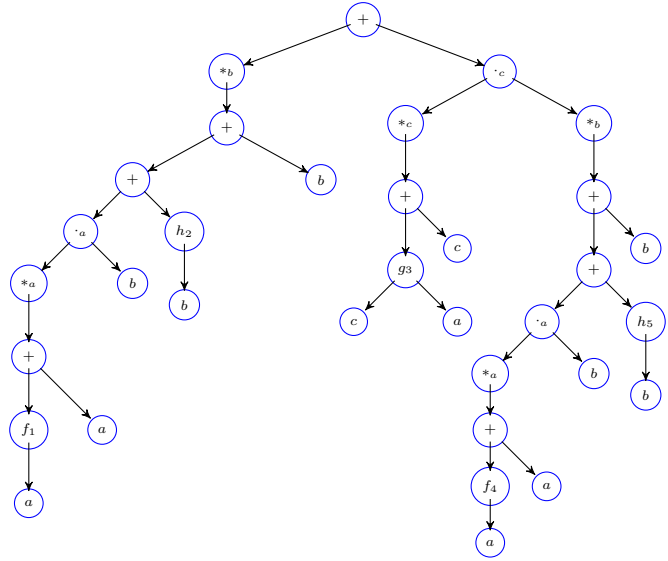
	
 The $\mathrm{ZPC}$-Structure associated with $\b\E=(f_1(a)^{*_a}\cdot_a b+ h_2(b))^{*_b}+g_3(c,a)^{*_c}\cdot_c (f_4(a)^{*_a}\cdot_a b+ h_5(b))^{*_b}$ is given in Figure~\ref{fig zpc e}.
\begin{figure}[H] 
	 \centerline{
	\begin{tikzpicture}[node distance=1.3cm,bend angle=30,transform shape,scale=.75]
		\node(1) {$+$};		
		\node[below right of=1,xshift=1.5cm] (3) {$\cdot_c$}; 
		\node[below left of=1,xshift=-1.5cm] (22) {$*_b$};			
		\node[below of=22,node distance=1cm] (221) {$+$};		
		\node[below left of=221,xshift=-.75cm] (5) {$+$};	 
		\node[below right of=221,xshift=.75cm] (222) {$b$};  
		\node[below right of=5] (16) {$h_2$};		 
		\node[below of=16] (4) {$\cancel{b}$};
		\node[below left of=5] (17) {$\cdot_a$};
		\node[below left of=17] (18) {$*_a$}; 
	 \node[below right of=17] (19) {$b$};
	  \node[below of=18] (181) {$+$};
	   \node[below of=181] (182) {$f_1$};
	   \node[right of=182] (184) {$a$};
	   \node[below of=182] (183) {$\cancel{a}$};
		\node[below right of=3,xshift=.75cm] (6) {$*_b$};
		\node[below of=6] (61) {$+$};
        \node[below right of=61] (62) {$b$}; 
		\node[below of=61] (51) {$+$};
		\node[below right of=51] (116) {$h_5$};		 
		\node[below of=116] (41) {$\cancel{b}$};
		\node[below left of=51] (117) {$\cdot_a$};
		\node[below left of=117] (118) {$*_a$}; 
	 \node[below right of=117] (119) {$b$};
	  \node[node distance=1cm,below of=118] (1181) {$+$};
	   \node[node distance=1cm,below of=1181] (1182) {$f_4$};
	  \node[node distance=1cm,below of=1182] (1183) {$\cancel{a}$};
	  \node[node distance=1cm,right of=1182] (1184) {$a$};
	\node[below left of=3,xshift=-.75cm] (31) {$*_c$};	
	\node[below of=31] (3121) {$+$};
	\node[below of=3121] (312) {$g_3$};
	\node[below right of=3121] (3122) {$c$};	 
	\node[below right of=312] (311) {$\cancel{a}$};
	\node[below left of=312] (313) {$\cancel{c}$};
	\node[inner sep=0pt,draw, circle, fit=(1184), color=blue] (eqa1) {};
	   \node[inner sep=0pt,draw, circle, fit=(51), color=blue] (eqa1) {};
	   \node[inner sep=0pt,draw, circle, fit=(116), color=blue] (eqa1) {};
	   \node[inner sep=0pt,draw, circle, fit=(117), color=blue] (eqa1) {};
	   \node[inner sep=0pt,draw, circle, fit=(118), color=blue] (eqa1) {};	   
	   \node[inner sep=0pt,draw, circle, fit=(119), color=blue] (eqa1) {};
	   \node[inner sep=0pt,draw, circle, fit=(1181), color=blue] (eqa1) {};
	   \node[inner sep=0pt,draw, circle, fit=(1182), color=blue] (eqa1) {};
	 \node[inner sep=0pt,draw, circle, fit=(1), color=blue] (eq1) {};
     \node[inner sep=0pt,draw, circle, fit=(181), color=blue] (eqf1) {};
     \node[inner sep=0pt,draw, circle, fit=(182), color=blue] (eqa1) {};
     \node[inner sep=0pt,draw, circle, fit=(184), color=blue] (eqa1) {};
	    \node[inner sep=0pt,draw, circle, fit=(22), color=blue] (eq2) {};	    
	    \node[inner sep=0pt,draw, circle, fit=(221), color=blue] (eq2) {};   
	    \node[inner sep=0pt,draw, circle, fit=(222), color=blue] (eq2) {};
	    \node[inner sep=0pt,draw, circle, fit=(3), color=blue] (eq3) {};
	    \node[inner sep=0pt,draw, circle, fit=(5), color=blue] (eq4) {};
	    \node[inner sep=0pt,draw, circle, fit=(16), color=blue] (eq6) {};
	    \node[inner sep=0pt,draw, circle, fit=(19), color=blue] (eq7) {};
	    \node[inner sep=0pt,draw, circle, fit=(18), color=blue] (eq8) {};	
	    \node[inner sep=0pt,draw, circle, fit=(17), color=blue] (eq9) {};
	    \node[inner sep=0pt,draw, circle, fit=(62), color=blue] (eq10) {};  
	    \node[inner sep=0pt,draw, circle, fit=(61), color=blue] (eq10) {};    
	    \node[inner sep=0pt,draw, circle, fit=(6), color=blue] (eq10) {};
	   \node[inner sep=0pt,draw, circle, fit=(3121), color=blue] (eq12) {};
	   \node[inner sep=0pt,draw, circle, fit=(312), color=blue] (eq12) {};
	   \node[inner sep=0pt,draw, circle, fit=(3122), color=blue] (eq12) {};
	   \node[inner sep=0pt,draw, circle, fit=(31), color=blue] (eq13) {};
	   \path[->]
	 	  (18) edge node[above left,pos=0.4] {} (181)
	      (181) edge node[above left,pos=0.4] {} (182)
	      (181) edge node[above left,pos=0.4] {} (184)
	      (1) edge node[above left,pos=0.4] {} (22)
	      (1) edge node[above right,pos=0.4] {} (3)
	      (61) edge node[left,pos=0.4] {} (62)
	      (6) edge node[left,pos=0.4] {} (61)
	      (61) edge node[left,pos=0.4] {} (51)
	      (22) edge node[left,pos=0.4] {} (221)
	      (221) edge node[left,pos=0.4] {} (222)
	      (221) edge node[left,pos=0.4] {} (5)
	      (3) edge node[right,pos=0.4] {} (6)
	      (5) edge node[above left,pos=0.4] {} (16)
	      (5) edge node[above left,pos=0.4] {} (17)
	      (17) edge node[left,pos=0.4] {} (18)
	      (17) edge node[left,pos=0.4] {} (19)
	      (3) edge node[left,pos=0.4] {} (31)
	      (31) edge node[left,pos=0.4] {} (3121)
	      (3121) edge node[left,pos=0.4] {} (3122)
	      (3121) edge node[left,pos=0.4] {} (312)
	      (51) edge node[left,pos=0.4] {} (116)
	      (51) edge node[left,pos=0.4] {} (117)
	      (117)edge node[left,pos=0.4] {} (118)
	      (117)edge node[left,pos=0.4] {} (119)
	      (118)edge node[left,pos=0.4] {} (1181)
	      (1181)edge node[left,pos=0.4] {} (1184)
	      (1181)edge node[left,pos=0.4] {} (1182)
	     ;
	     \path[-,color=black,dotted]
	      (1182)edge node[left,pos=0.4] {} (1183)
	      (16) edge node[left,pos=0.4] {} (4)
	      (182) edge node[above left,pos=0.4] {} (183)
	      (312) edge node[left,pos=0.4] {} (311)
	      (312) edge node[left,pos=0.4] {} (313)
	      (116) edge node[left,pos=0.4] {} (41)
	      ;
	      \path[->,color=red]	         
          (31) edge[bend right=15](6)
	      (118)edge[bend right=15] (119)
	      (61)edge[bend left=35] (6) 
	      (18) edge[bend right=35] (19)	
	      (181) edge[bend left=55] (18)	      
	      (221) edge[bend left=65] (22)
	      (3121) edge[bend left=65](31)
	      (1181)edge[bend left=55] (118) 
	      ;
	     \end{tikzpicture}
	 }
\caption{The $\mathrm{ZPC}$-Structure of $\b\E$}
	 \label{fig zpc e}
	\end{figure}
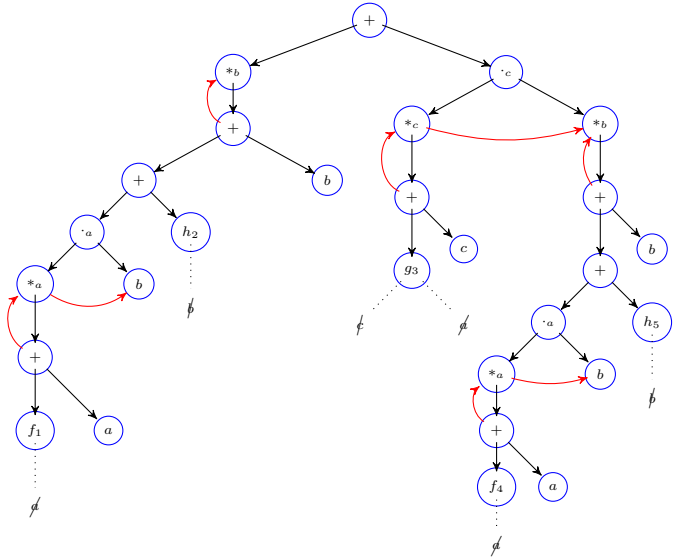

\begin{theorem}\label{theorem-zpc}
 The $\mathrm{ZPC}$-Structure associated with $\E$ can be computed in $O(|\E|)$ time and space complexity.
\end{theorem}
\begin{proof}
The first step of our Algorithm~\ref{zpc} consists of computing the sets $\firs{\f}$ for all subexpressions $\f$ of $\E$. 
The set $\firs{\f}$ is represented by an array where the entries are indexed by symbols of $\S_0$. 
The computation of all sets $\firs{\f}$ requires $O(|\E|)$ time and space complexity.\\  

Now that we have computed the sets $\firs{}$, the second step consists of the construction of the $\First$ Forest. 
Recall that this $\First$ Forest encodes the $\firstt{\f}$ sets for all subexpressions $\f$ of $\E$. 
 Therefore, the set $\firstt{\f}$ can be obtained by a prefix traversal of the syntax tree of $\E$ in $O(|\E|)$ time and space complexity.
\qed
\end{proof}
As each node $\nu_{\f}$ encodes $\firstt{\E_{\nu_{\f}}}$ we can state the following lemma.
\begin{lemma}\label{l1}
For a subexpression $\f$ of $\E$ the set $\firstt{\f}$ can be computed in $O(|\f|)$ time and space complexity.
\end{lemma}

For a regular expression $\E$, the following algorithm allows to compute the set $\Follow(\E,f_j,k)$ for a symbol $f_j\in\Sigma^{\E}_m$ and integers $1 \leq k\leq m$.\\

\begin{algorithm}[H]
  \KwIn{Regular Expression $\E$.}
  \KwOut{$\{\Follow(\E,f_j,k)\mid\ f_j \in\Sigma^{\E}_m, 1\leq k\leq m\}$.}
  \nl Calculate~{$\Follow(\E,f_j,k)=\las{\E}{f_j}{k}\uplus\Fw{\E}{f_j}{k}$}\\
     $(1.1)$ \For{ $\nu=\nu_{f_j}$ to $\nu_{\E}$ }{
            Compute~$\las{\E_{\nu}}{f_j}{k}$\;
       \textbf{end for}}
      $(1.2)$ Compute~$\Fw{\E}{f_j}{k}$\;
      \Return{$(\las{\E}{f_j}{k}\uplus\Fw{\E}{f_j}{k})$}
  \caption{Algorithm for the function $\Follow$ for $f_j$ and $k$}
  \label{algo:follow}
\end{algorithm}

For each step of the Algorithm~\ref{algo:follow} we will evaluate the complexity in time and in space.

We denote by $\displaystyle\sum_{f_j\in\Sigma_{>}}\mathrm{r}(f_j)$ the sum of all ranks of symbols 
$f_j\in\Sigma_{>}$.\\ 

\noindent\emph{Step~$1$: Computation of $\Fll{\E}{f_j}{k}=\las{\E}{f_j}{k}\uplus\Fw{\E}{f_j}{k}$}\\
We are interesting about the computation of the sets $\las{\E}{f_j}{k}$ and $\Fw{\E}{f_j}{k}$. \\

\noindent\emph{Step~$1.1$: Computation of sets $\las{\E_{\nu}}{f_j}{k}$}\\

At each node $\nu$ of the syntax tree $T_{\E}$ of $\E$, the set $\las{\E_{\nu}}{f_j}{k}$ is represented by an array where the entries are indexed by symbols of $\Sigma_0$.
 The computation of the set $\las{\E_{\nu}}{f_j}{k}$ requires an $O(|\E|)$ time and space complexity.\\  

\noindent\emph{Step~$1.2$: Computation of $\Fw{\E}{f_j}{k}$}\\

Now that $\las{\E_{\nu}}{f_j}{k}$ for all node $\nu$, such that $\nu_{f_j} \preccurlyeq \nu$, are computed, we can use the techniques outlined in the case of words to calculate the set $\Fw{\E}{f_j}{k}$.  
Indeed, our formulas given in the Proposition~\ref{x} for the computation of $\Fw{\E}{f_j}{k}$ are similar to that defined in the case of words~\cite{Brug,ZPC}. We have the same formulas so we can use the same algorithms used in the paper ~~\cite{ZPC} for the computation of the sets $\Follow$.
Therefore, the computation of  $\Fw{\E}{f_j}{k}$ can be done in $O(|\E|)$ time complexity.\\

We denote by $\mathrm{R}$ the maximal rank of symbols of $\Sigma$ appearing in $\E$. Recall that the alphabetic width $||\E||$, of a regular expression $\E$ is the sum of occurrences of symbols of a rank greater than $0$ appearing in $\E$ that is $||\E||=\sum_{f\in\Sigma_{>}}{|\E|}_f$. The size of the ranked alphabet $\Sigma$ is considered as constant. 
\begin{lemma}\label{lemcomplex}
Let $\E$ be a regular expression, $f_j$ be a symbol in $\Sigma^{\E}_m$ and $1 \leq k\leq m$ be two integers. The sets $\Follow(\E,f_j,k)$ for $1 \leq k\leq m$ can be computed in time $O(\mathrm{r}(f_j)\cdot|\E|)$.
\end{lemma}
As $(\displaystyle\sum_{f_j\in\Sigma_>}(\mathrm{r}(f_j)))$ is bounded by $(\mathrm{R}\cdot||\E||)$ we can state the following theorem.
\begin{theorem}\label{thmcomplex}
The sets $\Follow(\E,f_j,k)$ for all symbols $f_j$ in $\Sigma^{\E}_{>}$ and for all $1\leq k\leq \mathrm{r}(f_j)$ can be computed with an $O(\mathrm{R}\cdot||\E||\cdot|\E|)$ time complexity.
\end{theorem}
\subsection{Improving the computation of the function $\Follow$}
In this section we present a simple transformation of the regular expression $\E$ which allows us to efficiently compute the sets $\Follow$. For a subexpression $f_j(\E_1,\ldots,\E_m)$ of $\E$ and a symbol $a$ in $\displaystyle \bigcup^m_{i=1}\firs{\E_i}$ we associate an expression $\E^a_{f_j}$ obtained from $\E$ by replacing the subexpression $f_j(\E_1,\ldots,\E_m)$ by the expression $f_j(a)$. 
\begin{example}    
For the regular expression $\E=f(a+g(b),a+b+h(a))^{*_a}\cdot_b l(b)$. We get $\E^a_f=f(a)^{*_a}\cdot_b l(b)$ and $\E^b_f=f(b)^{*_a}\cdot_b l(b)$
\end{example}
For all subexpressions $f_j(\E_1,\ldots,\E_m)$ of $\E$ and for a symbol $a\in\displaystyle \bigcup^m_{i=1}\firs{\E_i}$, the following proposition gives the link between $\Follow(\E,f_j,k)$ and $\Follow(\E^a_{f_j},f_j,1)$.

\begin{proposition}\label{propFw}
Let $\E$ be a regular expression, $f_j(\E_1,\ldots,\E_m)$ be a subexpression of $\E$ and $1\leq k\leq m$ be two integers. 
 
 The set $\Follow(\E,f_j,k)$ can be computed as follows: 
 $$\Follow(\E,f_j,k)=\firstt{\E_k}\uplus \displaystyle\bigcup_{a\in \firs{\E_k}}\Follow(\E^a_{f_j},f_j,1)$$
\end{proposition}
\begin{proof}
\begin{sloppy}  
  
 For a subexpression $f_j(\E_1,\ldots,\E_n)$ of $\E$ and from Proposition~\ref{b}, the set $\Follow(\E,f_j,k)$ is of the form: 
$\Follow(\E,f_j,k)= ((((\First(\E_{\nu_0})\cdot_{op(\nu_1)} \First(\E_{\gamma(\nu_1)}))\cdot_{op(\nu_2)} \First(\E_{\gamma(\nu_2)}))
\dots \cdot_{op(\nu_{m})} \First(\E_{\gamma(\nu_m)}))$ where $\nu_{f_j}$ is the node of $T_{\E}$ labelled by $f_j$, $\nu_0$ is the $k\mbox{-}\mathrm{child}(\nu_{f_j})$, $\Gamma_{\nu_{f_j}}(\E)=\{\nu_1, \dots,\nu_m \}$ and for $1\leq i\leq m,~op(\nu_i)=c$ such that $\mathrm{sym}(\mathrm{father}(\nu_i)) \in \{\cdot_c,{*_c}\}$. 
 
  \begin{align*} 
\Follow(\E,f_j,k)&=((((\First(\E_{\nu_0})\cdot_{op(\nu_1)} \First(\E_{\gamma(\nu_1)}))\cdot_{op(\nu_2)} \First(\E_{\gamma(\nu_2)}))
\dots \cdot_{op(\nu_{m})} \First(\E_{\gamma(\nu_m)}))\\
&\ \ \mbox{ with } \E_{\nu_0}=\E_k\\
&=((((\firstt{\E_{\nu_0}}\uplus\firs{\E_{\nu_0}})\cdot_{op(\nu_1)} \First(\E_{\gamma(\nu_1)}))\cdot_{op(\nu_2)}
\dots \cdot_{op(\nu_{m})} \First(\E_{\gamma(\nu_m)}))\\
&\ \ \mbox{ with } \E_{\nu_0}=\E_k \mbox{ and } \First(\E_{\nu_0})=\firstt{\E_{\nu_0}}\uplus\firs{\E_{\nu_0}}\\
&=\firstt{\E_{\nu_0}}\uplus(((\firs{\E_{\nu_0}}\cdot_{op(\nu_1)} \First(\E_{\gamma(\nu_1)}))\cdot_{op(\nu_2)}
\dots \cdot_{op(\nu_{m})} \First(\E_{\gamma(\nu_m)}))\\
&=\firstt{\E_{\nu_0}}\uplus(((\displaystyle\bigcup_{a\in \firs{\E_{\nu_0}}}(a)\cdot_{op(\nu_1)} \First(\E_{\gamma(\nu_1)}))\cdot_{op(\nu_2)}\dots \cdot_{op(\nu_{m})} \First(\E_{\gamma(\nu_m)}))
\end{align*}
By using this last formula and the modifications: for all symbols $a \in\displaystyle \bigcup\firs{\E_{\nu_0}}$ we associate an expression 
$\E^a_{f_j}$ obtained from $\E$ by replacing the subexpression $f_j(\E_1,\ldots,\E_n)$ by the expression $f_j(a)$, then we have for $a\in\firs{\E_{\nu_0}}$:
\begin{align*} 
\Follow(\E^a_{f_j},f_j,1)=&(((a\cdot_{op(\nu_1)} \First(\E_{\gamma(\nu_1)}))\cdot_{op(\nu_2)}\dots \cdot_{op(\nu_{m})} \First(\E_{\gamma(\nu_m)}))
\end{align*} 
Therefore, for all symbols $a \in\displaystyle \bigcup\firs{\E_{\nu_0}}$: 
\begin{align*} 
\Follow(\E,f_j,k)=&\firstt{\E_{\nu_0}}\uplus \displaystyle \bigcup_{a \in\firs{\E_{\nu_0}}}\Follow(\E^a_{f_j},f_j,1)
\end{align*}
\end{sloppy}   
 \qed
\end{proof}

As the rank of the symbol $f_j$ in $\E^a_{f_j}$ is $1$ and by Lemma~\ref{lemcomplex}, the set $\Follow(\E^a_{f_j},f_j,1)$ can be computed in time $O(|\E|)$. 
 This step is considered as a preprocessing and is common to each symbol $a$ such that $a$ is in $\displaystyle\bigcap^n_{k=1} \firs{\E_k}$ for all $1\leq k\leq n$.

So, one can compute in first time the sets $\Follow(\E^a_{f_j},f_j,1)$ for all $a$ in $\displaystyle\bigcup^n_{k=1}\firs{\E_k}$ in $O(|\E|)$ time complexity. In the second time, from these sets and the set $\firstt{\E_k}$ we construct the set $\Follow(\E,f_j,k)$ using  formula of Proposition~\ref{propFw}. This second step can be performed in $O(|\E_k|+||\E||)$ time complexity. Indeed from Lemma~\ref{firstComput}, $\firs{\E_k}$ can be computed in time $O(|\E_k|)$  and the set $\displaystyle\bigcup_{a\in \firs{\E_k}} \Follow(\E^a_{f_j},f_j,1)$ can be constructed from the sets computed in the first step with an $O(||\E||)$ time complexity. 

\noindent As $\displaystyle(\sum^n_{k=1}|\E_k|)<|\E|$ and $(\mathrm{r}(f_j)\cdot||\E||)<|\E|$ and as the first step is performed once for all $k$, $1\leq k\leq n$ and for all $a\in\displaystyle\bigcap^n_{k=1} \firs{\E_k}$, then, we can state the following proposition.
\begin{proposition}    
Let $\E$ be a regular expression and $f_j$ be a symbol in $\Sigma^{\E}_{>}$. The set $\Follow(\E,f_j,k)$ for all $1\leq k\leq \mathrm{r}(f_j)$ can be computed with an $O(|\E|)$ time complexity. 
\end{proposition}    

Finally we can state the following theorem.

\begin{theorem}
Let $\E$ be a regular expression. The computation of the $\Follow$ sets for all symbol $f_j\in\Sigma^{\E}_{>}$ and $1\leq k\leq \mathrm{r}(f_j)$ can be done with an $O(|\E|\cdot||\E||)$ time complexity.
\end{theorem}
Our algorithm for the computation of the $\Follow$ sets can be used for the computation of the set of transition rules of the $k$-position automaton, the equation automaton~\cite{automate2,cie}, the $k$-c-continuation automaton~\cite{cie,arxiv} and the Follow automaton~\cite{arxiv}. 
\begin{remark} 
By analogy to the word case, we have chosen to don't consider the constant symbols $(\Sigma_0)$ in the alphabetic width of $\E$. For example for the regular expression $\E=f\underbrace{(a,\ldots,a)}_{a \ n\mbox{-times}}$, $||\E||=1$. However, in~\cite{automate2}, the alphabetic width is the number of occurrences of symbols of $\Sigma$ in $\E$, that is $||\E||=n+1$.   
\end{remark}

\section{Conclusion}
In this paper the notion of $k$-position tree automaton associated with the regular
tree expression has been recalled. This automaton is the generalization from words to trees of the position automaton introduced by Glushkov. We give an efficient algorithm that computes the $\Follow$ function from a regular expression $\E$ in $O(||\E||\cdot |\E|)$ time complexity. 

This algorithm for the computation of the $\Follow$ sets can be used for the computation of the set of transitions of the $k$-position, equation, $k$-c-continuation and $\Follow$ automata.  

\bibliographystyle{splncs_srt}
\bibliography{bibliography}
\end{document}